\documentclass[12pt,reqno]{amsart}
\pdfoutput=1 
\usepackage{amsmath,bbm}
\usepackage{latexsym}
\usepackage{amsfonts}
\usepackage{amssymb}
\usepackage{color}
\usepackage[table]{xcolor}

\usepackage{graphicx}
\usepackage{url}
\usepackage{enumerate}
\usepackage{tikz}
\usetikzlibrary{shadings, intersections, calc, plotmarks}
\usepackage{marginnote}
\usepackage{stackrel}
\usepackage{empheq}
\usepackage{tcolorbox}
\usetikzlibrary{arrows.meta,positioning}

\usepackage{geometry}
\usepackage{hyperref}
\hypersetup{
  colorlinks=true,
  linkcolor=blue!90!black,    
  citecolor=green!50!black,   
  urlcolor=blue!70!black,     
  linktoc=all,                
}
\usepackage{cleveref}

\xdefinecolor{tumblue}     {RGB}{0,101,189}
\xdefinecolor{tumgreen}    {RGB}{162,173,  0}
\xdefinecolor{tumred}      {RGB}{229, 52, 24}
\xdefinecolor{tumivory}    {RGB}{218,215,203}
\xdefinecolor{tumorange}   {RGB}{227,114, 34}
\xdefinecolor{tumlightblue}{RGB}{152,198,234}

\newtheorem{proposition}{Proposition}
\newtheorem*{proposition*}{Proposition}

\newtheorem{theorem}{Theorem}
\newtheorem*{theorem*}{Theorem}
\newtheorem{lemma}{Lemma}
\newtheorem{corollary}{Corollary}
\newtheorem*{corollary*}{Corollary}

\newcommand{\R}{\mathbbm{R}}

\newcommand{\C}{\mathbbm{C}}

\newcommand{\N}{\mathbbm{N}}

\newcommand{\Q}{\mathbbm{Q}}

\newcommand{\cT}{\mathcal{T}}

\newcommand{\1}{\mathbbm{1}}

\DeclareMathOperator{\Ran}{ran}
\DeclareMathOperator{\Span}{span}
\def\>{{\rangle}}
\def\<{{\langle}}
\def\Z{{\mathbb{Z}}}

\newcommand{\be}{\begin{equation}}
\newcommand{\ee}{\end{equation}}
\newcommand{\bea}{\begin{eqnarray}}
\newcommand{\eea}{\end{eqnarray}}

\newcommand{\Ex}[1]{\mathbbm{E}\left[#1\right]}

\newcommand{\rank}[1]{\mathrm{rank}\left[#1\right]}

\begin{document}

\title{ How many switches does two-Hamiltonian universality require?}

\author[Garbaravičius]{Matas Garbaravičius$^{1,2}$}
\address{$^1$ Mathematical Institute, University of Oxford, Oxford, United Kingdom}
\address{$^2$ Department of Mathematics, Technical University of Munich, Garching, Germany}

\author[Wolf]{Michael M. Wolf$^{2,3}$}
\address{$^3$ Munich Center for Quantum
Science and Technology (MCQST), Munich, Germany}
\email{m.wolf@tum.de}

\date{\today}

\begin{abstract} Any pair of sufficiently generic Hamiltonians on $\C^d$ generates the full unitary group, i.e., any $W\in U(d)$ can be synthesized to arbitrary precision by alternating between the two evolutions for suitably chosen durations. 
For any given pair of Hamiltonians, the minimal number of switches required to cover the entire group is at least $d$ and  unbounded in general. Despite this, we show 
that it scales linearly in $d$ in expectation. 
Assuming that the two spectra are sufficiently generic, we derive  constructive upper bounds on the minimal switching number as a function of the change-of-basis matrix between the two eigenbases. In particular, if this matrix has a column whose entries all have equal moduli, as for a complex Hadamard matrix, then $9d$ switches suffice to reach every unitary. To obtain these results, we construct a short exact  representation of any unitary as a product of unitaries that are alternately diagonal in two given bases. All bounds are extended to non-generic spectra and bounded evolution times.
\end{abstract}
\maketitle

\section{Introduction}\label{sec:intro}

Synthesizing a desired unitary evolution by switching between the evolutions generated by a few fixed Hamiltonians is a basic primitive of quantum control and quantum information processing. It underlies the observation that almost any pair
of Hamiltonians, or quantum gates, is universal \cite{lloyd1995almost}, and it is at the heart of quantum simulators \cite{Llo96} and variational quantum algorithms \cite{CAB+21}. 

The most elementary setting of this kind is the one where we can switch between two Hamiltonians $H_1$ and $H_2$ on $\mathbb{C}^d$. This is used for quantum approximate optimization algorithms \cite{FGG14,Llo18,MBZ20}, and is known under the name \emph{bang-bang control} in the context of control theory. If $iH_1$ and $iH_2$ generate the full Lie algebra $\mathfrak{u}(d)$, which is the case for generic pairs, then every $W\in U(d)$ can be written as
\begin{equation}\label{eq:word}
  W=e^{it_1\hat{H}_1}e^{it_2\hat{H}_2}\cdots e^{it_n\hat{H}_n},\qquad t_1,\ldots,t_n\geq 0,
\end{equation}
where the $\hat{H}_j$ alternate between $H_1$ and $H_2$, and $n$ can be
bounded independently of $W$ \cite{DAlessandro2002}. The question we address in this paper is how large $n$ has to be or, equivalently, how many switches $(n-1)$ are required in order to reach every unitary, at least to an arbitrarily good approximation.

\paragraph{\bf Related work and variations on the theme} For $d=2$, the question we address has a geometric reformulation: via
$U(2)/U(1)\simeq SU(2)/\{\pm 1\}\simeq SO(3)$, unitaries up to a global phase
correspond to rotations in $\mathbb{R}^3$, and $e^{itH_j}$ corresponds to a
rotation about an axis that is determined by the eigenbasis of $H_j$. The question
thus becomes: how many rotations about two fixed axes in $\mathbb{R}^3$ are required to
obtain an arbitrary rotation? It was shown in \cite{Low71,Low72, Hamada2014} that the answer is
$\lceil\pi/\theta\rceil+1$ if the two axes enclose an angle
$\theta\in(0,\pi/2]$. In particular, the required number of
switches cannot be bounded uniformly over all universal pairs, not even for $d=2$.

For $d>2$, to the best of our knowledge, no explicit upper bound is known for a general pair of Hamiltonians, although the problem has been discussed in various contexts \cite{JiaCapel2024,LloydEfficientImplementation2019,Lutz2026}. D'Alessandro \cite{DAlessandro2002}  derived a bound for generating a neighborhood of the identity. This bound, however, grows exponentially in the dimension of the group, i.e., exponentially in $d^2$, and its extension to the entire group relies on a compactness argument that does not yield an explicit bound. 

The minimal $n$ that allows us to represent every group element $W$ \emph{exactly} as in Equation~\eqref{eq:word} is called the \emph{order of generation}. Since every factor contributes only a single real parameter, a dimension count shows that it is at least
$d^2$. If it suffices to approximate every unitary arbitrarily well, considerably fewer switches may suffice. In fact, if the
eigenvalues of $H_j$ are generic in the sense that they are linearly independent over $\mathbb{Q}$, then, by Kronecker's theorem, $\{e^{itH_j}\,|\,t\geq 0\}$ is dense in the torus of all unitaries that are diagonal in the eigenbasis of $H_j$. Taking the eigenbasis of
$H_1$ as the computational basis and letting $U$ be the unitary whose columns are the eigenvectors of $H_2$, we denote these two tori by $\cT$ and $\cT_U:=U\cT U^*$. Assuming generic eigenvalues, every unitary can be approximated arbitrarily well with $n$ factors iff $(\cT\cup\cT_U)^n=U(d)$. Our question
therefore becomes: given two orthonormal bases,
how efficiently can any unitary be written as an alternating product of unitaries that are diagonal in one or the other basis? We denote the minimal number of
factors by $L(U)$. Since $\cT$ and $\cT_U$ are $d$-dimensional and share the global phases, a dimension count now gives 
$L(U)\geq d+1$, so that at least $d$ switches are required.

If $U$ is the discrete Fourier transform $F$, a related question has been studied in linear optics \cite{LLM21,girouard2025near}, where diagonal unitaries are implemented by phase masks and Fourier transforms by fixed passive optical elements. For even $d$ it was shown that every unitary can be decomposed into $2d+5$ phase
masks interleaved with $2d+4$ Fourier transforms.
Since $F^2$ is a permutation matrix that squares to the identity, one can write $FDF=(FDF^*)F^2$ and move all factors
$F^2$ to the right end of such a decomposition. Since the remaining fixed
permutation can be absorbed into the target unitary, the number of phase masks
bounds $L(F)$. In particular, $L(F)\le 2d+5$ for even $d$.
 Similar (but incomparable) representations in terms of diagonal matrices and Fourier transforms follow from the Sinkhorn normal form
for unitary matrices \cite{idel2015sinkhorn} and from a factorization into circulant and diagonal matrices \cite{huhtanen2015factoring}.\vspace*{5pt}

\paragraph{\bf Overview of results} 
In this paper, we consider arbitrary pairs of bases with a change-of-basis matrix $U$ and derive explicit upper bounds on $L(U)$. On the level of Hamiltonian universality, as considered in \Cref{eq:word}, this means that we assume sufficiently generic spectra and ask how many switches are required to reach a dense subset of $U(d)$. In \Cref{sec:nongeneric} we will discuss the implications of the results in scenarios where the assumption on the spectra is not met and/or there is a bound imposed on the evolution time.

Our main tool is an exact projective representation of any unitary as a product of at most $d-1$ generalized Householder reflections, in the spirit of \cite{HouseholderUN}, which we choose adaptively so that each of them admits a short representation in terms of elements of $\cT$ and $\cT_U$. This leads to a general upper bound on $L(U)$ in terms of the moduli
of the entries of $U$ (\Cref{thm:main1}). If $U$ has a column whose entries all have equal moduli, as is the case for every complex Hadamard matrix, we obtain
$L(U)<8.2\,d-5$ (\Cref{cor:flat8}). This is only a factor $4.1$ weaker than the bound for the Fourier transform \cite{girouard2025near}, which exploits the full structure of $F$, whereas our bound only relies on the moduli of a single column of $U$. 

In \Cref{thm:graph} we show that $L(U)<\infty$ iff a bipartite graph that encodes the non-zero entries of $U$ is connected, and that in this case $L(U) = {\mathcal O}(d^2/\mu)$, where $\mu$ is a lower bound on the moduli of the relevant entries of $U$.

Finally, for a Haar-random change-of-basis matrix, we show that $\mathbb{E}[L(U)]\leq Cd$ for a
universal constant $C$ (\Cref{thm:HaarLU}), so that, together with the lower bound above, the required number of switches scales linearly in $d$ in expectation.
 
All results are formulated in a slightly more general setting in which the second Hamiltonian may act non-trivially only on a subset $Q$ of its eigenvectors, which covers, for instance, Hamiltonians of low rank. For the linear scaling in \Cref{thm:HaarLU}, $|Q|=2$ non-zero eigenvalues suffice.

\section{Preliminaries}\label{sec:prelim}
We assume $d\in\N$ and $d\geq 2$ throughout. For a unit vector $v\in\C^d$ and an angle $\varphi\in[0,2\pi]$, we define the \emph{Householder reflection}\footnote{This is also known as \emph{generalized Householder reflection} or \emph{complex reflection}.}
$$R_v(\varphi):=\1+\big(e^{i\varphi}-1\big)vv^* \in U(d).$$
Note that this leaves the entire hyperplane $v^\perp$ unchanged and thus can be thought of as a reflection about this hyperplane.

The following elementary property can be found in the literature (e.g., Thm. 3.4 in \cite{Mackey2003}). We provide a short proof for completeness. 
\begin{lemma}\label{lem:Rxy}
    Let $x,y\in\C^d$ be distinct with equal norms and $v:=(y-x)/\|y-x\|$. Then there is a $\varphi\in[0,2\pi]$ s.t. $R_v(\varphi)x=y$. 
\end{lemma}
\begin{proof}
    Since $y-x\propto v$, the vectors $x$ and $y$ must have the same projection on $v^\perp$. Denoting this projection by $z$, we can write $x=z+\alpha v$ and $y=z+\beta v$ for some $\alpha,\beta\in\C$.  The assumed equality $\|x\|=\|y\|$ then forces $|\alpha|=|\beta|$, so that there is a $\varphi\in[0,2\pi]$ with $\alpha e^{i\varphi}=\beta$. 
\end{proof}
A useful consequence is the following:
\begin{lemma}\label{lem:rankreduction}
    Let $W\in U(d)$ be a unitary, and $v\in\Ran(W-\1)$ a unit vector. Then there is a $\varphi\in[0,2\pi]$ s.t.
    $$\rank{R_v(\varphi)W-\1}=\rank{W-\1}-1 .$$
\end{lemma}
\begin{proof}
    Choose $z$ s.t. $(W-\1)z=v$. Since $z$ and $Wz$ have equal norms and difference proportional to $v$, \Cref{lem:Rxy} guarantees the existence of a $\varphi$ with $(R_v(\varphi)W-\1)z=0$. Normality of $W$ gives $\ker(W-\1)=\ker(W^*-\1)=\Ran(W-\1)^\perp$. Hence, on the subspace $\ker(W-\1)\subseteq v^\perp$, the map $R_v(\varphi)$ acts as the identity. Therefore, $\ker(R_v(\varphi)W-\1) \supseteq \ker(W-\1)$. This inclusion is strict due to the vector $z$ so that the rank of $R_v(\varphi)W-\1$ must be strictly smaller than the rank of $W-\1$. On the other hand, the difference of the ranks can be at most one since the difference of the operators has rank one.
\end{proof}
A rank reduction similar to the one in \Cref{lem:rankreduction} is used in \cite{HouseholderUN} to construct a decomposition of an arbitrary unitary $W\in U(d)$  into Householder reflections. In fact, iterating \Cref{lem:rankreduction} shows that there exist $\varphi_1,\ldots,\varphi_d\in\R$ and $v_2,\ldots,v_d\in \C^d$ such that
\begin{equation}
    W=e^{i\varphi_1}R_{v_2}(\varphi_2)\cdots R_{v_d}(\varphi_d).\label{eq:Hhproduct}
\end{equation}
The idea in the present paper is to choose the $R_{v_j}
(\varphi_j)$'s in an adaptive way (and thus different from \cite{HouseholderUN}), so that each of them admits a minimal-cost decomposition in terms of the available transformations.

\section{Alternating diagonal representation}\label{sec:altdiag}
We fix a unitary $U\in U(d)$, denote by $\cT\subseteq U(d)$ the diagonal unitaries and by $\cT_{U}:=U\cT U^*$ the unitaries that are diagonal in the basis given by the column vectors of $U$. Our main results will only use a subset of $\cT_{U}$ defined by $\cT_{U,Q}:=\{V\in\cT_{U}|(U^*VU)_{k,k}\neq 1\Rightarrow k\in Q\}$ for a subset $Q\subseteq[d]:=\{1,\ldots, d\}$.

The aim of this section is to find good upper bounds on 
\begin{equation}\label{eq:LUdef}
    L_Q(U):=\min\big\{n\;|\;n\in\N,(\cT\cup \cT_{U,Q})^n=U(d)\big\},\quad \text{where}\quad\min\emptyset:=\infty,
\end{equation}
and $(\cT\cup \cT_{U,Q})^n$ is the set of products $D_1 \cdots D_n$ with $D_i\in\cT\cup\cT_{U,Q}$. That is, $L_Q(U)$ is the least $n$ such that every element of $U(d)$ is a \emph{word} of length at most $n$ in the \emph{alphabet} $\cT\cup \cT_{U,Q}$.
  Since both $\cT$ and $\cT_{U,Q}$ are groups, and thus closed under multiplication, the shortest product representing a given element in $U(d)$ is always one with factors alternating between $\cT$ and $\cT_{U,Q}$. 
  
  The relation to the Hamiltonian universality problem outlined in \Cref{eq:word} is the content of the following proposition.
Since a simultaneous unitary conjugation of $H_1$ and $H_2$ does not affect the question, we may assume that $H_1$ is diagonal or, equivalently, express everything in an eigenbasis of $H_1$.

\begin{proposition}\label{prop:NL}
Let $H_1=\operatorname{diag}(\lambda_1,\ldots,\lambda_d)$ and
$H_2=U\operatorname{diag}(\eta_1,\ldots,\eta_d)U^*$ with $\lambda,\eta\in\R^d$ and
$U\in U(d)$, and set $Q:=\{k\,|\,\eta_k\neq0\}$. Denote by
$N(H_1,H_2)\in\N\cup\{\infty\}$ the smallest $n$ such that the products in
\Cref{eq:word}, with $t_1,\ldots,t_n\geq0$ and either Hamiltonian first, form a
dense subset of $U(d)$. If $(\lambda_i)_{i\in[d]}$ and $(\eta_k)_{k\in Q}$ are each
linearly independent over $\Q$, then 
 \begin{equation}
          N(H_1,H_2) =L_{Q}(U).
      \end{equation}
\end{proposition}
\emph{Remark: } Note that linear independence over $\Q$ holds with probability one if the eigenvalues are drawn randomly from an absolutely continuous distribution.
\begin{proof}
Let $T_j:=\{e^{itH_j}\,|\,t\geq0\}$ for $j=1,2$. By Kronecker's theorem (e.g.,
\cite[Ch.~XXIII]{HardyWright}), linear independence of $\lambda_1,\ldots,\lambda_d$
over $\Q$ implies that $\overline{T_1}=\cT$. Likewise, the assumption on $(\eta_k)_{k\in Q}$ gives
$\overline{T_2}=\cT_{U,Q}$.

As $T_1$ and $T_2$ are closed under multiplication and contain $\1$, the set of
products in \Cref{eq:word} with either Hamiltonian first is $(T_1\cup T_2)^n$. The multiplication map from the compact set
$(\cT\cup\cT_{U,Q})^{\times n}$ to $U(d)$ is continuous, so it maps the dense subset
$(T_1\cup T_2)^{\times n}$ onto a dense subset of its compact image
$(\cT\cup\cT_{U,Q})^n$. Hence
$\overline{(T_1\cup T_2)^n}=(\cT\cup\cT_{U,Q})^n$, so the products in
\Cref{eq:word} are dense in $U(d)$ iff $(\cT\cup\cT_{U,Q})^n=U(d)$. Comparing with \Cref{eq:LUdef} yields $N(H_1,H_2)=L_Q(U)$.
\end{proof}

Inspired by \Cref{eq:Hhproduct}, we will build up $U(d)$ from Householder reflections, which will be chosen so that they themselves can be written as products of only a few elements from $\cT$ and $\cT_{U,Q}$. 

 Let us call a column of $U$ whose index is contained in $Q$ a $Q$-\emph{column}.
 If $v$ is a $Q$-column of $U$, say $U e_j=v$ with $j\in Q$, then $R_v(\varphi)\in\cT_{U,Q}$ since $R_v(\varphi)=U(\1+(e^{i\varphi}-1)e_j e_j^*)U^*$ and $(U^*R_v(\varphi)U)_{k,k}=1+(e^{i\varphi}-1)\delta_{j,k}$. Similarly, for any $D\in\cT$: $R_{Dv}(\varphi)\in\cT\cT_{U,Q}\cT$. That is, for any vector whose moduli match the ones of a $Q$-column of $U$, the corresponding Householder reflections can be implemented at the cost of only a single element from $\cT_{U,Q}$. If the moduli are different, the following can be used:
\begin{lemma}\label{lem:ampamp}
    Let $u,v\in\C^d$ be unit vectors and let $s$ be the largest constant s.t. $\forall i:s|v_i|\leq |u_i|$. Assume $s>0$ and define\vspace*{-4pt} \begin{equation}m(s):=\left\lceil\frac{\pi}{4 \arcsin s}-\frac12\right\rceil,\quad n:=2m(s)+1.\end{equation} For any $\varphi\in[0,2\pi]$ there exist $\theta_1,\ldots,\theta_n\in[0,2\pi]$ and $D_0,\ldots,D_n\in\cT$ s.t.
    \begin{equation}\label{eq:RvRseq}
        R_v(\varphi)=D_0 R_u(\theta_1)D_1 R_u(\theta_2)\cdots D_{n-1}R_u(\theta_n)D_n.
    \end{equation}
\end{lemma}
\begin{proof} W.l.o.g., we can assume that all entries of $v$ are non-negative, since all phases can be absorbed into diagonal unitaries. Similarly, we may use $Du$ for any $D\in\cT$ instead of $u$.
    We will construct a finite sequence of vectors in $\C^d$,
    $$Du=u(0)\longmapsto u(1)\longmapsto \cdots \longmapsto u(m)=v,  $$ starting at $Du$ and ending at $v$, so that, up to diagonal unitaries, the difference of any neighboring pair $u(j)$, $u(j+1)$ is proportional to $u$.  \Cref{lem:Rxy} then guarantees the existence of a corresponding sequence of $m$ Householder reflections $R_u(\theta_j)$ that map $u$ step-by-step to $v$ when interleaved with suitable diagonal unitaries. Denoting the resulting map by $G:u\mapsto v$, we can write $R_v(\varphi)=G R_u(\varphi) G^*$. Unraveling $G$, we obtain \Cref{eq:RvRseq} with $n=2m+1$ Householder reflections of the form $R_u(\theta_j)$.

In order to construct the sought sequence, we define
\begin{equation}
    z_\pm(\alpha):=\sin(\alpha)v\pm i\cos(\alpha) w,\quad\text{for }\quad \alpha\in[0,\pi/2],
\end{equation}
where the entries of $w$ are chosen non-negative and such that for some $\beta\in[0,\pi/2]$, to be chosen later, we have $\forall i:|z_\pm(\beta)_i|=|u_i|$. The latter implies that $$w_i=\frac{\sqrt{|u_i|^2-v_i^2\sin(\beta)^2}}{\cos(\beta)},$$ which requires $\sin(\beta)\leq s$. The equal-moduli condition means we can set $u(0):=z_+(\beta)$. Before we construct the other points in the sequence, note that the trigonometric identities 
    \begin{align} \sin(\alpha+2\beta) - \sin(\alpha) &= 2\cos(\alpha+\beta) \sin(\beta),\quad \text{and}\nonumber  \\ \cos(\alpha+2\beta) + \cos(\alpha) &= 2\cos(\alpha+\beta)\cos(\beta)\quad\text{imply}\nonumber\\
    \label{eq:zrec}
    z_-(\alpha+2\beta)-z_+(\alpha)&=2\cos(\alpha+\beta)z_-(\beta).\end{align}
Since $z_-(\beta)$ has the same moduli as $u$, \Cref{lem:Rxy} together with \Cref{eq:zrec} implies that for any $\alpha$ the map $z_+(\alpha)\mapsto z_-(\alpha+2\beta)$ can be implemented by a single Householder reflection $R_u(\theta)$ sandwiched by diagonal unitaries. The same is true for $z_+(\alpha)\mapsto z_+(\alpha+2\beta)$ since the components of $z_+$ and $z_-$ have equal moduli (the diagonal unitary realizing $z_-(\alpha + 2\beta) \mapsto z_+(\alpha+2\beta)$ can be absorbed into one of those sandwiching $R_u(\theta)$). Hence, we can set $u(j):=z_+\big((2j+1)\beta\big)$. This yields $u(m)=v$ if $(2m+1)\beta=\pi/2$, which is consistent with $\sin(\beta)\leq s$ only if $$\frac{\pi}{4\arcsin(s)}-\frac12\leq m.$$ Choosing the smallest such $m\in\N$, and the corresponding value for $\beta$, then completes the proof for any $s<1$. For $s=1$, however, $v$ and $u$ have the same moduli so that the statement of the Lemma follows immediately by diagonal conjugation.
\end{proof}
For any $u\in\C^d$ and $k\in\{1,\ldots,d\}$ let $u_{[k]}\in\R^k$ be the vector that contains the $k$ smallest moduli of the coordinates of $u$ and denote by $\|u_{[k]}\|$ its Euclidean norm. The following will help us to identify the vector $v$ to which \Cref{lem:ampamp} should be applied:
\begin{lemma}\label{lem:uk}
    Let $k\in\{1,\ldots,d\}$, $u\in\C^d$ and $V\subseteq\C^d$ be a $k$-dimensional subspace. There is a unit vector $v\in V$ satisfying
    \begin{equation}
        \forall i:\quad\|u_{[k]}\|\cdot |v_i|\leq |u_i|.
    \end{equation}
\end{lemma}
\begin{proof}
    Consider a maximizer of the norm $\|x\|$ over the compact convex set $C:=\{x\in V|\forall i:|x_i|\leq|u_i|\}$. We will show that at least $k$ of these inequalities must be equalities: suppose equality holds only for $l<k$ coordinates and  $P:V\rightarrow \C^l$ projects onto them. Then $P$ must have a non-trivial kernel, i.e., there is a non-zero $y\in V$ with $Py=0$. Hence, $x\pm\epsilon y\in C$ for sufficiently small $\epsilon>0$. By the parallelogram law, their average squared norm exceeds that of $x$ since\vspace*{-6pt}
    $$\|x+\epsilon y\|^2+\|x-\epsilon y\|^2=2\|x\|^2+2\epsilon^2\|y\|^2 > 2\|x\|^2,$$
    contradicting the maximality of $\|x\|$. So any maximizer $x\in C$ of $\|x\|$ satisfies $\|x\|\geq\|u_{[k]}\|$.  Then $v:=x/\|x\|$ satisfies
    $\|u_{[k]}\|\cdot |v_i|\leq |x_i|\leq |u_i|   $ for all $i$.
\end{proof}
Now we have all the ingredients to prove the following:
\begin{theorem}\label{thm:main1} Let $u\in\C^d$ be a unit vector such that every $R_u(\varphi)$ can be written as a product of elements of $\cT\cup\cT_{U,Q}$ using at most $c$ elements from $\cT_{U,Q}$. Then
    \begin{equation}\label{eq:LUbound}
        L_{Q}(U)\leq1+2c\sum_{k=1}^{d-1}\Big[2m\big(\|u_{[k]}\|\big)+1\Big].
    \end{equation}
\end{theorem}
\emph{Remark: }Note that if $u$ is a $Q$-column of $U$, then $c=1$.
\begin{proof}
    Let $W\in U(d)$ be the target that we want to write as a product of elements from $\cT\cup\cT_{U,Q}$ with a number of factors that is bounded by the r.h.s. of \Cref{eq:LUbound}. 
    By multiplying $W$ with a phase, we can ensure that $W_0:=e^{i\varphi_0}W$  has an eigenvalue $1$. Then define recursively $W_{l}:=R_{v_{l}}(\varphi_{l})W_{l-1}$, where $v_{l}\in V_l:=\Ran(W_{l-1}-\1)$ is chosen by \Cref{lem:uk} and $\varphi_l$ by \Cref{lem:rankreduction} in order to obtain $\dim(V_{l+1})=\dim(V_l)-1$. Following \Cref{lem:ampamp},  $R_{v_{l}}(\varphi_l)$ costs at most $2m\big(\|u_{[\dim(V_l)]}\|\big)+1$ applications of $R_u(\cdot)$ and thus at most $c$ times as many  elements from $\cT_{U,Q}$. After $K\leq(d-1)$ steps, this iteration finally yields $W_K=\1$ and therefore $R_{v_K}(\varphi_K)\cdots R_{v_1}(\varphi_1)e^{i\varphi_0}=W^*$. Taking the adjoint of this equation, and using that generally $R_v(\varphi)^*=R_v(-\varphi)$, yields a representation of $W$ as a product of elements in $\cT\cup\cT_{U,Q}$ that uses at most $$N:=c\sum_{k=1}^{d-1}\Big[2m\big(\|u_{[k]}\|\big)+1\Big]$$ elements of $\cT_{U,Q}$. Merging neighboring factors from the same group and inserting identities
where necessary yields $W\in(\mathcal{T}\mathcal{T}_{U,Q})^N\mathcal{T}
$ and therefore $L_{Q}(U)\leq 2N+1$.
\end{proof}
In order to obtain a better understanding of the scaling behavior of the bound in \Cref{{eq:LUbound}}, we provide an estimate of the involved sum in terms of a lower bound on the moduli of the entries of $u$:

\begin{lemma}
    Let $u\in\C^d$ be a unit vector whose moduli are bounded from below by $|u_i|\geq\nu>0$ for all $i$. Then the following bounds hold:
    \begin{equation}\label{eq:nubounds}
        \sum_{k=1}^{d-1}\Big[2m\big(\|u_{[k]}\|\big)+1\Big]\leq\begin{cases}
            3(d-1)+\left(\frac{\pi}{3\nu}\right)^2\\
            2(d-1)+\frac{\pi\sqrt{d-1}}{\nu}
        \end{cases}
    \end{equation}
\end{lemma}
\begin{proof} The assumption on $u$ implies that $\|u_{[k]}\|\geq\nu\sqrt{k}$, so that monotonicity of $m$ yields $m(\|u_{[k]}\|)\leq m(\nu\sqrt{k})$.

    For the first bound, we use that for any $x>0$ we can write  $\lceil x\rceil=\sum_{j\geq 1}\mathbf{1}_{x>j-1}$ and that $\pi/(4\arcsin s)-\tfrac12>j-1$ holds iff $s<\sin(\pi/(4j-2))$. Combining this, we get
    \begin{eqnarray}
        \sum_{k=1}^{d-1}m\big(\nu\sqrt{k}\big) &=&\sum_{j\geq 1}\Big|\Big\{k\in\{1,\ldots,d-1\}\Big|\;k\nu^2<\Big[\sin\frac{\pi}{4j-2}\Big]^2\Big\}\Big|\nonumber\\
        &\leq& (d-1)+\frac{1}{\nu^2}\underbrace{\sum_{j\geq 2} \Big[\sin\frac{\pi}{4j-2}\Big]^2}_{\leq\pi^2/18}\nonumber
    \end{eqnarray}
    In the inequality to the second line, we count every $k$ in the case of $j=1$ and for $j\geq 2$ we use that the number of $k$'s below any $y$ is at most $y$. The remaining sum can be bounded by $\pi^2/18$ by using that $\sin^2(x)\leq x^2-x^4/3+2x^6/45$ and summing up the resulting series\footnote{The inequality can be shown by integrating $\cos t\le 1$ six times over
$[0,x]$ and then using that $\sin^2(x)=(1-\cos 2x)/2$.}. Inserting this into the l.h.s. of \Cref{eq:nubounds} yields the first bound.

    To arrive at the second bound, we use $\lceil x\rceil\leq x+1$ and $\arcsin(s)\geq s$. So  
    \begin{eqnarray}
        \sum_{k=1}^{d-1}m\big(\nu\sqrt{k}\big) &\leq& \sum_{k=1}^{d-1}\left( \frac{\pi}{4\nu\sqrt{k}}+\frac12\right)\nonumber\\
        &=& \tfrac12(d-1)+\frac{\pi}{4\nu}\sum_{k=1}^{d-1}\frac{1}{\sqrt{k}}.\nonumber
    \end{eqnarray}
    Finally, the sum can be bounded via $\sum_{k=1}^{d-1}k^{-1/2}\leq 1+\int_1^{d-1}x^{-1/2} dx\leq 2\sqrt{d-1}$.
\end{proof}
From here, the smallest bound is obtained for unitaries that have a $Q$-column whose entries have equal moduli -- such as normalized complex Hadamard matrices or, more specifically, the discrete Fourier transform:
\begin{corollary}\label{cor:flat8}
    If $U\in U(d)$ has a $Q$-column $u\in\C^d$ s.t. $|u_i|=1/\sqrt{d}$ for all $i$, then 
    \begin{equation}
        L_{Q}(U)\leq \left(6+\frac{2\pi^2}{9}\right) d -5<8.2\, d-5.
    \end{equation}
\end{corollary}
\begin{proof}
    Using the first bound in \Cref{eq:nubounds} and inserting it into \Cref{eq:LUbound} with $c=1$ proves the claim.
\end{proof}

If every $Q$-column of $U$ has a zero entry, then the bounds obtained from \Cref{thm:main1} with $c=1$ are void. One way to obtain a bound that is faithful in the sense that it is finite iff $L_{Q}(U)$ is, is to consider the coupling graph associated with $U$:

To each $U\in U(d)$ we associate a weighted bipartite graph $G$ whose $d+|Q|$ vertices are all the rows and $Q$-columns of $U$, with row $i$ connected to column $j\in Q$  iff $U_{ij}\neq 0$. The modulus $|U_{ij}|$ is then the weight assigned to this edge.
\begin{theorem}\label{thm:graph}
    For any $U\in U(d)$ we have:
    \begin{enumerate}[(i)]
        \item If $L_{Q}(U)<\infty$, then $G$ is connected.
        \item Conversely, if $G$ is connected, then $L_{Q}(U)<\infty$. More specifically, if $\mu>0$ is the minimal edge weight of a connected spanning subgraph of $G$, then
        \begin{equation}\label{eq:LUmu}
            L_{Q}(U)\leq 1+d(d-1)(2m(\mu)+1)\leq 1+\frac{3d (d-1)}{\mu}.
        \end{equation}
    \end{enumerate}
\end{theorem}
\begin{proof} \emph{(i) :}
         Suppose $G$ is disconnected. Let $I$ and $J$ be the sets of rows and columns of the component that contains the first row. Then $U$ maps $S_J:=\Span\{e_j|j\in J\}$ into $S_I:=\Span\{e_i|i\in I\}$ and disconnectedness of $G$ is reflected in the fact that $\dim{S_I}<d$. The subspace $S_I$ is clearly invariant under $\cT$, but also
under $\cT_{U,Q}$: every $V\in\cT_{U,Q}$ is of the form $V=UDU^*$ with $D\in\cT$ and
$D_{k,k}=1$ for $k\notin Q$, so that for all $x\in\mathbb{C}^d$
\[
  Vx-x=\sum_{j\in Q}\big(D_{j,j}-1\big)\langle u_j,x\rangle\,u_j .
\]
For $x\in S_I$, the terms with $j\in J$ lie in $S_I$ and those with
$j\in Q\setminus J$ vanish, so $Vx\in S_I$. Since not every unitary leaves a given proper subspace invariant, $(\cT\cT_{U,Q})^n\neq U(d)$ for all $n\in\N$, contradicting $L_{Q}(U)<\infty$. 

        \emph{(ii) :} Denote by $\Gamma$ a connected spanning subgraph of $G$ whose minimal edge weight is $\mu$. Lemma II.1 in \cite{Duniversality} together with connectedness of $\Gamma$ guarantees that every $W\in U(d)$ can be brought into diagonal form by multiplying with at most $K:=d(d-1)/2$ two-level unitaries, each of which couples a pair of rows that are linked by a column in $\Gamma$ (i.e., the rows  have distance $2$ in $\Gamma$). We will use the resulting decomposition of $W$ and now bound the `cost' of its factors:

        Let $B\in U(d)$ be a non-diagonal two-level building block of $W$. Then there are two levels $a\neq b$ s.t. $B$ and $B^*$ leave $S:=\Span\{e_a,e_b\}$ invariant and act as the identity on $S^\perp$. Moreover, there is a $j\in Q$ s.t. $|U_{aj}|,|U_{bj}|\geq\mu$. For the unit vector $v\propto (B e_a-e_a)\in S$, \Cref{lem:Rxy} guarantees the existence of a $\varphi$ s.t. $R_v(\varphi) e_a=Be_a$. Then $D:=R_v(\varphi)^* B$ acts as the identity on $S^\perp$ and leaves $e_a$ and thus $\Span\{e_b\}$ invariant. So $D\in\cT$ and $B=R_v(\varphi) D$. 

        If we set $u:=Ue_j$, then $|u_i|\geq \mu |v_i|$ for all $i$ so that by \Cref{lem:ampamp}, $R_v(\varphi)$ and therefore $B$ can be expressed as a product in $\cT\cup\cT_{U,Q}$ with at most $2m(\mu)+1$ uses of $\cT_{U,Q}$. Taking into account that there are at most $K$ of these factors yields the first inequality in \Cref{eq:LUmu}. For the second inequality note first that since at least two moduli of $u$ are bounded below by $\mu$, normalization of $u$ demands $\mu\leq1/\sqrt{2}$. The inequality then follows from $\arcsin \mu \geq\mu$ and $\lceil x\rceil\leq x+1$ via $2m(\mu)+1\leq \pi/(2\mu)+2\leq(\pi/2+\sqrt{2})/\mu<3/\mu$.
\end{proof}

In the remaining part of this section, we will bound the expectation value of $L(U)$ w.r.t. the Haar measure on $U(d)$.

\begin{theorem}\label{thm:HaarLU} Assume $|Q|>1$. With respect to the Haar measure on $U(d)$:
\begin{equation}\label{eq:linExbound}
    \Ex{L_{Q}(U)}\leq Cd-11,\qquad\text{with}\quad C:=12+18\pi\Gamma\big(\tfrac14\big)^{4/3}.
\end{equation}
\end{theorem}
\emph{Remarks:} One could use  $\Gamma(1/4)< 2\pi/\sqrt{3}$. The constant $C$ is not optimized.
\begin{proof} Let $a,b\in\C^d$ denote two distinct $Q$-column vectors of $U$ and define $$\Theta:=\frac1d\sum_{i=1}^d (d|a_i b_i|)^{-3/2}.$$ In the first part of the proof, we will construct a suitable vector $u$ from $a$ and $b$ that has `cost' $c=3$  in \Cref{thm:main1} and show that 
\begin{equation}\label{eq:Houtlineineq}
    L_{Q}(U)\leq 12d-11+3\pi\sum_{k=1}^{d-1}\|u_{[k]}\|^{-1}\leq d\big(12+18\pi\Theta^{2/3}\big)-11.
\end{equation}
Note that the first inequality is merely \Cref{eq:LUbound} for $c=3$ combined with the estimate $m(s)\leq\pi/(4s)+1/2$.
In the second part of the proof, we will then show that $\Ex{\Theta^{2/3}}$ can be bounded by a constant (i.e., independent of $d$) such that \Cref{eq:Houtlineineq} leads to \Cref{eq:linExbound}.\vspace*{3pt}

Using that the entries of $U$ are nonzero with probability one, we can define 
\begin{equation}
    u:=R_b(\pi/3)x,\quad \text{with}\quad x_i:=|a_i|\frac{b_i}{|b_i|}.
\end{equation}
Since $b$ is a $Q$-column of $U$, we have $R_b(\pi/3)\in\cT_{U,Q}$ and as the coordinates of $x$ have the same moduli as those of $a$, which is also a $Q$-column of $U$, we know that $R_x(\varphi)\in \cT\cT_{U,Q}\cT$. So $R_u(\varphi)=R_b(\pi/3)R_x(\varphi)R_b(\pi/3)^*$ is a product of elements of $\cT_{U,Q}\cup\cT$ that uses at most $c=3$ elements of $\cT_{U,Q}$. Hence, the first inequality in \Cref{eq:Houtlineineq} is justified. Direct expansion of $u_i$ leads to 
\begin{eqnarray}
    |u_i|^2 &=& \big(|a_i|-\langle b,x\rangle |b_i|\big)^2 + \langle b,x\rangle |a_i b_i|\nonumber \\
    &\geq& \langle b,x\rangle |a_i b_i|\quad\text{with}\quad \langle b,x\rangle=\sum_i|a_i b_i|.\label{eq:bxsum}
\end{eqnarray}
    Regarding $\langle b,x\rangle$ as something that is roughly one,  \Cref{eq:bxsum} means that we bound $|u_i|$ essentially by the geometric mean of the moduli of two  entries in $U$.

    Next, we work towards an upper bound on the sum in \Cref{eq:Houtlineineq}. For any subset $I\subseteq \{1,\ldots, d\}$ of $k$ elements, we can use H\"older's inequality with exponents $5/3$ and $5/2$ to obtain
\begin{eqnarray*}
            k = \sum_{i \in I} |a_ib_i|^{3/5} |a_ib_i|^{-3/5} &\leq & \left(\sum_{i \in I} |a_ib_i|\right)^{3/5} \left(\sum_i |a_ib_i|^{-3/2}\right)^{2/5}\\
            \text{so that }\ \sum_{i\in I}|a_i b_i| &\geq& k^{5/3}\bigg(\sum_{i=1}^d |a_i b_i|^{-3/2}\bigg)^{-2/3}\hspace{-4pt}=\bigg(\frac{k}{d}\bigg)^{5/3}\Theta^{-2/3}.
        \end{eqnarray*}
Applying this twice, with \Cref{eq:bxsum}, and choosing $I$ accordingly,  yields
\begin{equation*}
              \|u_{[k]}\|^2 = \sum_{i \in I} |u_i|^2 \geq \bigg(\sum_{i=1}^d|a_ib_i|\bigg)\bigg(\sum_{i \in I} |a_ib_i|\bigg) \geq \bigg(\frac{k}{d}\bigg)^{5/3}\Theta^{-4/3}.
\end{equation*}
Taking the square root of the inverse,  we can finally bound the desired sum
\begin{equation}\label{eq:sumukinvtheta}
    \sum_{k=1}^{d-1}{\|u_{[k]}\|^{-1}} \leq d^{5/6} \Theta^{2/3}\sum_{k=1}^{d-1} k^{-5/6}\leq 6d\Theta^{2/3},
\end{equation}
where the last step follows from the integral bound $\sum_{k=1}^{d-1} k^{-5/6}\leq \int_0^{d-1} x^{-5/6} dx\leq 6d^{1/6}$. Using \Cref{eq:sumukinvtheta} then completes the proof of \Cref{eq:Houtlineineq}.

Now we turn to the probabilistic part. Taking the expectation value in \Cref{eq:Houtlineineq} reduces to computing $\Ex{\Theta^{2/3}}$, which can be bounded from above by $ \Ex{\Theta}^{2/3}$ using Jensen's inequality. Furthermore, as marginal distributions of all row vectors in $U$ are the same due to the invariance of the Haar measure under multiplying with a fixed permutation matrix, we have 
\begin{equation}\label{eq:Thetaa1b1}
    \Ex{\Theta}=d^{-3/2}\Ex{|a_1b_1|^{-3/2}}.
\end{equation}
If $g=(g_1,\ldots,g_d)$ is a vector of independent standard complex Gaussian random variables\footnote{That is, each $g_i$ has density $\pi^{-1}e^{-|z|^2}$ on $\C$ and $|g_i|^2$ has density $e^{-x}$ for $x>0$.} and $r:=\|g\|^2$, then $w:=g/\sqrt{r}$ has the same distribution as any row of a Haar unitary (cf. \cite[Sec.3]{Mastrodonato2007}). Moreover, $w$ and $r$ are independent (\cite[p.38]{multivariatestat}) so that, in particular,
\begin{equation}
    \Ex{r^{-3/2}}\Ex{|w_1 w_2|^{-3/2}}=\Ex{|g_1 g_2|^{-3/2}}=\Ex{|g_1|^{-3/2}}^2.
\end{equation}
Identifying w.l.o.g. $a_1=w_1$ and $b_1=w_2$, we can bound \Cref{eq:Thetaa1b1}  using
\begin{eqnarray*}
    \Ex{|g_1|^{-{3/2}}} &=&\int_0^\infty x^{-3/4} e^{-x} dx\ =\ \Gamma\big(\tfrac14\big),\quad\text{and}\\
    \Ex{r^{-3/2}}&\geq& \Ex{r}^{-3/2}=\Big(\sum_{i=1}^d \Ex{|g_i|^2}\Big)^{-3/2}=\ d^{-3/2},
\end{eqnarray*} where the inequality is again Jensen's. Combining pieces, we get
\begin{equation}\label{eq:ThetaGamma}
\Ex{\Theta^{2/3}}\leq\Ex{\Theta}^{2/3}\leq\Gamma\big(\tfrac14\big)^{4/3}. 
\end{equation}
Now we can take the expectation value in \Cref{eq:Houtlineineq}, apply \Cref{eq:ThetaGamma}, and finally obtain \Cref{eq:linExbound}.
\end{proof}

\section{Non-generic spectra and bounded times}\label{sec:nongeneric}

In this section, we are going to relax two assumptions that we made in order to reduce the Hamiltonian universality problem to the alternating diagonal unitary problem: the assumptions of generic spectra and unbounded time-intervals. We  do not attempt to address the resulting problem in general, but rather analyze what the results obtained so far imply in this context. The overall idea is the following: if the eigenvalues of the Hamiltonians have linear dependencies over $\Q$ and/or we impose a bound on how large the times $t_1,\ldots,t_n$ in \Cref{eq:word} are allowed to be, then the two tori $\cT$ and $\cT_{U,Q}$ may have to be replaced with smaller, reachable sets.  The differences between the full tori and their replacements combined with the derived switching-number bounds then allow us to derive approximation results with  error estimates.   

Throughout this section, we consider $H_1=\operatorname{diag}(\lambda)$ and $H_2=U\operatorname{diag}(\eta)U^*$ with arbitrary $\lambda,\eta\in\R^d$ and $U\in U(d)$. Let $\mathcal{W}_n$ denote the products in \Cref{eq:word} with non-negative evolution times and either Hamiltonian first. We first describe the orbit closures and obtain approximation bounds for non-generic spectra. Then we incorporate bounded evolution times.\vspace*{5pt}

\paragraph{\bf Orbit closures and approximation}
The integer relations among the eigenvalues determine the orbit closures. For $x\in\R^d$, define
\begin{align*}
  D_x&:=\operatorname{diag}(e^{2\pi ix_1},\ldots,e^{2\pi ix_d}),\\
  \Lambda(\lambda)&:=\{k\in\Z^d\,|\,\langle k,\lambda\rangle=0\},\\
  \cT(\lambda)&:=\{D_x\,|\,x\in\R^d,\ \langle k,x\rangle\in\Z\ \forall k\in\Lambda(\lambda)\}.
\end{align*}

\begin{lemma}\label{lem:closure}
  For every $\lambda\in\R^d$, the closure of $\{e^{it\operatorname{diag}(\lambda)}\,|\,t\geq0\}$ is $\cT(\lambda)$.
\end{lemma}
\begin{proof}
  The version of Kronecker's theorem in \cite[Thm.~1]{Fischer1998}, applied with  $\Theta_j=2\pi x_j$, shows that $D_x$ belongs to the closure of $\{e^{it\operatorname{diag}(\lambda)}\,|\,t\in\R\}$ iff $\langle k,x\rangle\in\Z$ for every $k\in\Lambda(\lambda)$. Hence, this closure equals $\cT(\lambda)$.  To restrict to non-negative times, use sequential compactness to choose integers $m_1<m_2<\cdots$ such that $e^{im_j\operatorname{diag}(\lambda)}$ converges and $t_j:=m_{j+1}-m_j\to\infty$. Then $e^{it_j\operatorname{diag}(\lambda)}\to\1$, so for every $s\geq0$,
  \[
    e^{-is\operatorname{diag}(\lambda)}=\lim_j e^{i(t_j-s)\operatorname{diag}(\lambda)}.
  \]
  Since $t_j-s\geq0$ for all sufficiently large $j$, the two orbit closures coincide.
\end{proof}
From this we see that the two relevant tori become $\cT(\lambda)$ and $U\cT(\eta)U^*$. Combining this with the reasoning used in the proof of \Cref{prop:NL} we obtain  
\begin{equation}\label{eq:ovWn}
    \overline{\mathcal{W}_n}=\big(\cT(\lambda)\cup U\cT(\eta)U^*\big)^n.
\end{equation}
To obtain approximation bounds, we measure how far each orbit closure is from the corresponding full torus. With $\|\cdot\|$ denoting the operator norm, set
\[
  \Delta_Q(\lambda):=\max_{D\in\cT_{\1,Q}}\ \min_{D'\in\cT(\lambda)}\|D-D'\|,
  \qquad \Delta(\lambda):=\Delta_{[d]}(\lambda).
\]
In particular, $\Delta(\lambda)$ is the Hausdorff distance between $\cT(\lambda)$ and $\cT$, and vanishes iff $\Lambda(\lambda)=\{0\}$.

\begin{theorem}\label{thm:approx}
  Let $Q\subseteq[d]$ and let $\ell\in\N$ satisfy $\ell\geq L_Q(U)$. Then for all $n\geq\ell$ and $W\in U(d)$,
  \begin{equation}\label{eq:approx}
    \inf_{V\in\mathcal{W}_n}\|W-V\|\;\leq\;\Big\lceil\frac\ell2\Big\rceil\big(\Delta(\lambda)+\Delta_Q(\eta)\big).
  \end{equation}
\end{theorem}
\begin{proof}
  Write $W=W_1\cdots W_\ell$ with factors alternating between $\cT$ and $\cT_{U,Q}$, inserting identities if needed. Replace each factor by a closest element of $\cT(\lambda)$ or $U\cT(\eta)U^*$, respectively. By \Cref{eq:ovWn}, the resulting product lies in $\overline{\mathcal{W}_\ell}$. Each type of factor occurs at most $\lceil\ell/2\rceil$ times, and the product error is at most the sum of the factor errors, which proves the bound for $n=\ell$. The remaining cases follow by adding identity factors.
\end{proof}

Since the operator norm distance between two unitaries is at most $2$, short integer relations can make this bound void. Indeed, for any $m\in\Lambda(\lambda)\setminus\{0\}$:
\begin{equation}\label{eq:Deltasinbound}
\Delta(\lambda)\geq2\sin\!\left(\frac{\pi}{2\|m\|_1}\right).
\end{equation}
To see this, choose $x_i=\operatorname{sgn}(m_i)/(2\|m\|_1)$. For every $D_y\in\cT(\lambda)$, we have $\langle m,x-y\rangle\in\frac12+\Z$, so some coordinate of $x-y$ has distance at least $1/(2\|m\|_1)$ from the nearest integer. We can then obtain a lower bound on $\Delta(\lambda)$ by using \[\|D_x-D_y\|=\max_{j}\big|e^{2\pi i x_j}-e^{2\pi i y_j}\big|=\max_{j} 2\big|\sin\big(\pi(x_j-y_j)\big)\big| .\] Here, only the distance of $|x_j-y_j|$ from $\Z$ is relevant, which then leads to the bound in \Cref{eq:Deltasinbound}.
Since $\|W-V\|\leq2$, \Cref{eq:approx} can be nontrivial only if every non-zero relation among $\lambda$ has $\ell_1$-norm larger than $\lceil\ell/2\rceil$.
A similar argument applies to $\Delta_Q(\eta)$. In particular, repeated eigenvalues of $H_1$ and repeated non-zero eigenvalues of $H_2$ lead to a trivial bound.\vspace*{5pt}

\paragraph{\bf Bounded evolution times}
Bounds on the evolution times require a quantitative measure of the separation from integer relations. For $a\in\R^r$ with $r\geq1$ and a positive integer $M$, define
\[
  \delta_M(a):=\min_{\substack{m\in\Z^r\\0<\|m\|_\infty\leq M}}|\langle m,a\rangle|.
\]
Thus, $\delta_M(a)>0$ means that $a$ has no non-zero integer relation of height at most $M$. The next theorem combines this spectral information with the word length bounds from \Cref{sec:altdiag}.

\begin{theorem}[Finite-time approximation]\label{thm:finite-time}
  Let $Q:=\{k\,|\,\eta_k\neq0\}$ be non-empty, set $\eta_Q:=(\eta_k)_{k\in Q}$, and let $\ell\in\N$ satisfy $\ell\geq L_Q(U)$. Suppose that $\delta_M(\lambda)>0$ and $\delta_{M'}(\eta_Q)>0$ for positive integers $M,M'$. For $T\geq2\pi/\delta_M(\lambda)$ and $T'\geq2\pi/\delta_{M'}(\eta_Q)$, every $W\in U(d)$ admits a product $V\in\mathcal{W}_\ell$ whose $H_1$- and $H_2$-evolution times lie in $[0,T]$ and $[0,T']$, respectively, such that
  \begin{equation}\label{eq:finite-time}
    \|W-V\|\;\leq\;2\pi\Big\lceil\frac\ell2\Big\rceil
    \left(\sqrt{\frac{\pi^2d}{16M^2}+\frac{\pi d}{T\delta_M(\lambda)}}
    +\sqrt{\frac{\pi^2|Q|}{16(M')^2}+\frac{\pi|Q|}{T'\delta_{M'}(\eta_Q)}}\right).
  \end{equation}
  The total evolution time is at most $\lceil\ell/2\rceil(T+T')$.
\end{theorem}
\begin{proof}
  We first approximate a single diagonal unitary. Let $a\in\R^r$, $\delta_M(a)>0$, and $D=D_x\in U(r)$, using the notation $D_x$ also in dimension $r$. Write $\|u\|_{\R/\Z}:=\min_{j\in\Z}|u-j|$. Chen's quantitative Kronecker theorem \cite{Che00}, in the form of \cite[Thm.~5.1]{GM16}, applies with frequencies $a_i/(2\pi)$, weights $1$, and height bounds $M$. The separation parameter is $\delta_M(a)/(2\pi)$, so for every $T\geq2\pi/\delta_M(a)$ there is a $t\in[0,T]$ with
  \[
    \sum_{i=1}^r\Big\|\frac{ta_i}{2\pi}-x_i\Big\|_{\R/\Z}^2
    \;\leq\;\frac{\pi^2r}{16M^2}+\frac{\pi r}{T\delta_M(a)}.
  \]
  The inequality $|e^{2\pi iu}-e^{2\pi iv}|\leq2\pi\|u-v\|_{\R/\Z}$ gives
  \begin{equation}\label{eq:spectral-time}
    \|D-e^{it\operatorname{diag}(a)}\|\;\leq\;
    2\pi\sqrt{\frac{\pi^2r}{16M^2}+\frac{\pi r}{T\delta_M(a)}}.
  \end{equation}
  Now write $W=W_1\cdots W_\ell$ with factors alternating between $\cT$ and $\cT_{U,Q}$. Apply \Cref{eq:spectral-time} to the first type with $a=\lambda, r=d$ and to the second type, after conjugating by $U^*$ and restricting to $Q$, with $a=\eta_Q, r=|Q|$. The entries outside $Q$ remain equal to $1$. Each type occurs at most $\lceil\ell/2\rceil$ times. Summing the factor errors as in \Cref{thm:approx} proves \Cref{eq:finite-time}, and adding up the evolution times gives the stated duration bound.
\end{proof}

For $M,M'\geq 2$, the choices
\begin{equation}\label{eq:time-choice}
  T=\frac{16M^2}{\pi\delta_M(\lambda)},\qquad
  T'=\frac{16(M')^2}{\pi\delta_{M'}(\eta_Q)}
\end{equation}
give the simpler bound
\[
  \|W-V\|\leq\Big\lceil\frac\ell2\Big\rceil\frac{\pi^2}{\sqrt2}
  \left(\frac{\sqrt d}{M}+\frac{\sqrt{|Q|}}{M'}\right).
\]
If the entries of $\lambda$ and those of $\eta_Q$ are each linearly independent over $\Q$, then arbitrarily large $M,M'$ are allowed. Hence, any prescribed accuracy can be reached with $\ell$ factors and the explicit time bounds in \Cref{eq:time-choice}. Note that these are uniform in the target $W$. This refines the density statement in \Cref{prop:NL} by controlling the evolution times as well as the number of switches.

\section*{Acknowledgments}
MMW acknowledges support from the German-Israeli Foundation for Scientific Research and Development (GIF) under grant no. I-3014-304.6/2026 and from the 
 Isaac Newton Institute for Mathematical Sciences, Cambridge (supported by EPSRC grant no. EP/Z000580/1) for  hospitality during the programme \emph{Mathematics of many-body entanglement} where work on this paper was undertaken.

\bibliographystyle{halpha}
\bibliography{Alt}{}\vspace*{15pt}

@article{Mackey2003,
author = {Mackey, D. Steven and  Mackey, Niloufer and Tisseur, Françoise},
journal = {Electron. J. Linear Algebra},
pages = {106-145},
title = {Structured tools for structured matrices.},
url = {http://eudml.org/doc/123080},
volume = {10},
year = {2003},
}

@article{HouseholderUN,
  title = {Engineering of arbitrary $\mathrm{U}(N)$ transformations by quantum {Householder} reflections},
  author = {Ivanov, P. A. and Kyoseva, E. S. and Vitanov, N. V.},
  journal = {Phys. Rev. A},
  volume = {74},
  issue = {2},
  pages = {022323},
  numpages = {8},
  year = {2006},
  publisher = {American Physical Society},
  doi = {10.1103/PhysRevA.74.022323},
  url = {https://link.aps.org/doi/10.1103/PhysRevA.74.022323}
}

@article{Duniversality,
  title = {Criteria for exact qudit universality},
  author = {Brennen, Gavin K. and O'Leary, Dianne P. and Bullock, Stephen S.},
  journal = {Phys. Rev. A},
  volume = {71},
  issue = {5},
  pages = {052318},
  numpages = {7},
  year = {2005},
  publisher = {American Physical Society},
  doi = {10.1103/PhysRevA.71.052318}
}

@book{multivariatestat,
publisher = {John Wiley and Sons, Ltd},
title = {Aspects of Multivariate Statistical Theory},
author= {Robb J. Muirhead},
year = {1982}
}

@article{Mastrodonato2007,
author={Mastrodonato, Christian
and Tumulka, Roderich},
title={Elementary Proof for Asymptotics of Large {Haar}-Distributed Unitary Matrices},
journal={Letters in Mathematical Physics},
year={2007},
volume={82},
number={1},
pages={51-59}
}

@misc{LloydEfficientImplementation2019,
	title={Efficient implementation of unitary transformations}, 
	author={Seth Lloyd and Reevu Maity},
	year={2019},note={arXiv:1901.03431},
	archivePrefix={arXiv},
	primaryClass={quant-ph},
	url={https://arxiv.org/abs/1901.03431}, 
}

@article{Hamada2014,
    author = {Hamada, Mitsuru},
    title = {The minimum number of rotations about two axes for constructing an arbitrarily fixed rotation},
    journal = {Royal Society Open Science},
    volume = {1},
    number = {3},
    pages = {140145},
    year = {2014}
}

@article{DAlessandro2002,
	title = {Uniform finite generation of compact {Lie} groups},
	journal = {Systems \& Control Letters},
	volume = {47},
	number = {1},
	pages = {87-90},
	year = {2002},
	author = {Domenico D'Alessandro}
}

@phdthesis{Lutz2026,
  TITLE = {{Contr{\^o}le en temps minimal : applications {\`a} l'informatique quantique et {\`a} la r{\'e}sonance param{\'e}trique}},
  AUTHOR = {Lutz, Killian},
  URL = {https://theses.hal.science/tel-05699839},
  NOTE = {AEX-MARCQ},
  SCHOOL = {{Strasbourg University}},
  YEAR = {2026},
  TYPE = {Theses},
  HAL_ID = {tel-05699839},
  HAL_VERSION = {v1},
}

@article{JiaCapel2024,
   title={A generic quantum {Wielandt's} inequality},
   volume={8},
   journal={Quantum},
   publisher={Verein zur Forderung des Open Access Publizierens in den Quantenwissenschaften},
   author={Jia, Yifan and Capel, Angela},
   year={2024}, pages={1331} 
}

@article{huhtanen2015factoring,
  title={Factoring matrices into the product of circulant and diagonal matrices},
  author={Huhtanen, Marko and Per{\"a}m{\"a}ki, Allan},
  journal={Journal of Fourier Analysis and Applications},
  volume={21},
  number={5},
  pages={1018--1033},
  year={2015},
  publisher={Springer}
}

@article{LLM21, author={L{\'o}pez-Pastor, V{\'i}ctor and  Lundeen, J. S. and Marquardt, F.},
  title={Arbitrary optical wave evolution with {F}ourier transforms and phase masks},
  journal={Opt. Express}, volume={29}, number={23}, pages={38441--38450}, year={2021}}

@article{girouard2025near,
  title={Near-optimal decomposition of unitary matrices using phase masks and the discrete {Fourier} transform},
  author={Girouard, Vincent and Quesada, Nicol{\'a}s},
  journal={Journal of the Optical Society of America B},
  volume={43},
  number={3},
  pages={A66--A73},
  year={2025},
  publisher={Optica Publishing Group}
}

@article{lloyd1995almost,
  title={Almost any quantum logic gate is universal},
  author={Lloyd, Seth},
  journal={Phys. Rev. Lett.},
  volume={75},
  number={2},
  pages={346},
  year={1995},
  publisher={APS}
}

@article{Llo96, author={S. Lloyd}, title={Universal quantum simulators},
  journal={Science}, volume={273}, pages={1073--1078}, year={1996}}

@article{idel2015sinkhorn,
  title={Sinkhorn normal form for unitary matrices},
  author={Idel, Martin and Wolf, Michael M},
  journal={Linear Algebra and its Applications},
  volume={471},
  pages={76--84},
  year={2015},
  publisher={Elsevier}
}

@article{CAB+21, author={M. Cerezo and A. Arrasmith and R. Babbush and others},
  title={Variational quantum algorithms}, journal={Nat. Rev. Phys.}, volume={3},
  pages={625--644}, year={2021}}

@misc{FGG14, author={E. Farhi and J. Goldstone and S. Gutmann},
  title={A quantum approximate optimization algorithm}, note={arXiv:1411.4028}, year={2014}}

@misc{Llo18, author={S. Lloyd}, title={Quantum approximate optimization is computationally universal},
  note={arXiv:1812.11075}, year={2018}}

@article{MBZ20, author={M. E. S. Morales and J. D. Biamonte and Z. Zimbor\'as},
  title={On the universality of the quantum approximate optimization algorithm},
  journal={Quantum Inf. Process.}, volume={19}, pages={291}, year={2020}}

@article{Low71, author={F. Lowenthal}, title={Uniform finite generation of the rotation group},
  journal={Rocky Mountain J. Math.}, volume={1}, number={4}, pages={575--586}, year={1971}}

@article{Low72, author={F. Lowenthal}, title={Uniform finite generation of {SU}(2) and {SL}(2,{R})},
  journal={Canad. J. Math.}, volume={24}, number={4}, pages={713--727}, year={1972}}

@book{HardyWright, author={G. H. Hardy and E. M. Wright},
  title={An Introduction to the Theory of Numbers}, edition={4th},
  publisher={Oxford University Press}, year={1975}}

@article{GM16, author={S. M. Gonek and H. L. Montgomery},
  title={Kronecker's approximation theorem}, journal={Indag. Math. (N.S.)},
  volume={27}, number={2}, pages={506--523}, year={2016}}

@article{Che00, author={Y.-G. Chen}, title={The best quantitative {K}ronecker's theorem},
  journal={J. London Math. Soc. (2)}, volume={61}, number={3}, pages={691--705}, year={2000}}

@Article{Fischer1998,
author={Fischer, Alexander},
title={Approximation of almost periodic functions by periodic ones},
journal={Czechoslovak Mathematical Journal},
year={1998},
volume={48},
number={2},
pages={193-205}
}

\end{document}